\documentclass[twocolumn]{autart}

\usepackage{graphicx}
\usepackage{amsmath,amssymb,amsfonts,bm}
\usepackage{booktabs}
\usepackage[authoryear]{natbib}
\usepackage{tikz}
\graphicspath{{figures/}{figures_eps/}{./}}

\DeclareMathOperator{\diag}{diag}
\DeclareMathOperator{\col}{col}
\DeclareMathOperator{\atanh}{atanh}

\newcommand{\R}{\mathbb{R}}
\newcommand{\one}{\mathbf{1}}
\newcommand{\abs}[1]{\left|#1\right|}
\newcommand{\pos}[1]{\left[#1\right]_{+}}

\begin{document}
\journal{Arxiv} 

\begin{frontmatter}

\title{Prescribed Performance Leader-Following Consensus with Event-Based Broadcasting} 




\author[Athena]{Athanasios K. Gkesoulis\corauthref{cor1}}
\ead{a.gkesoulis@athenarc.gr},
\author[NTUA]{Haris E. Psillakis}
\ead{hpsilakis@central.ntua.gr},
\author[Thessaly]{George C. Karras}
\ead{gkarras@uth.gr},
\author[Patras]{Charalampos P. Bechlioulis}
\ead{chmpechl@upatras.gr}

\corauth[cor1]{Corresponding author.}

\address[Athena]{Robotics Institute, Athena Research Center,
15125 Marousi, Greece}

\address[NTUA]{School of Electrical and Computer Engineering,
National Technical University of Athens, 15780 Athens, Greece}

\address[Thessaly]{Department of Informatics and Telecommunications,
University of Thessaly, 35100 Lamia, Greece}

\address[Patras]{Department of Electrical and Computer Engineering,
University of Patras, 26504 Patras, Greece}


\begin{abstract}                          
This paper develops a prescribed performance leader-following consensus
protocol for directed networks under event-based broadcasting. Each node
schedules its transmissions using only its current state and its last broadcast
value, while followers use held neighbor and leader samples between
broadcasts. To accommodate jumps in sampled neighborhood errors caused by
asynchronous receptions, a decaying receiver-side performance correction term
temporarily modifies the prescribed performance function. For a directed network containing a spanning tree rooted at a time-varying leader with bounded trajectory and
derivative, we prove existence and uniqueness of complete non-Zeno solutions,
boundedness of all closed-loop signals, and explicit prescribed bounds for the
sampled and continuous neighborhood disagreements and the leader-following error. These bounds
quantify the effects of the communication topology and broadcast thresholds.
Simulations verify the prescribed performance bounds and illustrate the
event-based communication mechanism.
\end{abstract}

\end{frontmatter}

\section{Introduction}
\label{sec:introduction}


Consensus requires networked agents to agree on a common state or trajectory
using locally exchanged information \citep{ren2008distributed}. 
Despite extensive research on this problem, reducing communication while guaranteeing prescribed transient and steady-state performance remains a fundamental challenge. 
Event-based sampling and event-triggered feedback were initially developed for
single systems \citep{Astrom2002Lebesgue,Tabuada2007} and later extended to
distributed multi-agent control \citep{Dimarogonas2012}. Whether these methods
reduce communication depends on the information used by the trigger: conditions
involving current relative states require continuous access to neighboring
states, whereas self-triggered methods predict future update times from
information available at the preceding event
\citep{Dimarogonas2012,Ding2018Survey}. An alternative is transmitter-local
broadcasting, where each agent schedules its transmissions using locally
available information \citep{WangLemmon2011}. In particular,
\citet{Seyboth2013Broadcasting} use the difference between the current state
and its last broadcast for average consensus. Related communication-saving designs address dynamic average consensus
\citep{Kia2015DynamicAverage} and consensus over directed graphs
\citep{Li2021}, however, without prescribing transient and
steady-state leader-following bounds.


Prescribed performance control (PPC) addresses this requirement by
constraining an error within a designer-selected time-varying envelope and mapping it through a barrier transformation \citep{4639441}. PPC has subsequently been extended to distributed synchronization \citep{Bechlioulis2017Synchronization}, and the combination of PPC and event-triggered mechanisms has attracted increasing attention. 
Event-triggered PPC has been studied for pinning synchronization, fuzzy
consensus with quantization, simultaneous state and controller-output
triggering, prescribed-settling-time consensus, and stochastic bipartite
consensus
\citep{Fan2020PPCEvent,Wang2021PPCEvent,Zhang2022FuzzyPPC,
Chen2022SettlingPPC,Ren2023BipartitePPC}. Related developments treat
virtual-leader networks, output synchronization, neural containment, unknown
saturation and control directions, adaptive performance boundaries,
fault-tolerant consensus, and switched containment
\citep{Yue2023PPCEvent,Peng2024CooperativePPC,Tang2025NeuralContainment,
Li2026EventDrivenPPC,Qiu2026SelfAdjustingPPC,Zheng2026FixedTimePPC,
Hou2026ContainmentPPC}. These results, however, do not simultaneously provide
prescribed leader-following performance over a leader-rooted directed graph
and schedule every transmission, including that of the leader, using only
the transmitter's current state and its last broadcast value.

Motivated by this gap, we develop a prescribed performance leader-following
consensus protocol under event-based broadcasting. Each node transmits when
the difference between its current state and its last broadcast reaches a
local threshold. No neighbor-, leader-, edge-, relative-state-, or
control-input-dependent quantity is used by the trigger. Followers evolve
continuously using their current states and held incoming samples. Because an
asynchronous reception may increase a sampled neighborhood error
instantaneously, each receiver employs a decaying performance modification term
that temporarily modifies its performance function. For leader-rooted directed
graphs, the resulting closed-loop system has unique, complete, non-Zeno
solutions, all signals remain bounded, and prescribed sampled, continuous, and
leader-following error bounds are guaranteed.


The closest event-triggered PPC schemes differ in both graph assumptions and
information requirements. The leaderless designs in
\citep{Hu2024PPCEvent,Hou2024UnifiedPPC,Xu2025PPCEvent} use fixed undirected
graphs: \citet{Hu2024PPCEvent} considers an undirected tree or a connected
undirected graph, whereas \citet{Hou2024UnifiedPPC} and
\citet{Xu2025PPCEvent} assume connected undirected graphs.
\citet{Wang2024LeaderPerformance} considers a directed graph with a spanning
tree rooted at the leader, but assumes that the leader is continuously
available. Moreover, the event-triggered implementations in
\citep{Hu2024PPCEvent,Hou2024UnifiedPPC,Xu2025PPCEvent} use edge or relative
measurements, while \citet{Wang2024LeaderPerformance} monitors neighboring
states. Their self-triggered extensions avoid current neighbor measurements
but retain event calculations involving edge dynamics, relative quantities,
neighboring control signals, or network-dependent bounds. Here, each
transmitter uses only its current state and last broadcast, and each receiver
uses its current state and held incoming samples.

The contributions are threefold: (i) a transmitter-local broadcasting
architecture is developed for prescribed performance leader-following
consensus over leader-rooted directed graphs, including event-based leader
transmissions; (ii) a receiver-side performance modification term is introduced to accommodate reception-induced jumps while preserving continuous follower evolution between broadcasts; and (iii) explicit topology- and threshold-dependent prescribed
performance bounds are established that quantify the designer-chosen communication--performance tradeoff. To the
best of the authors' knowledge, this is the first distributed PPC result over
directed graphs in which every transmission is scheduled exclusively from the
transmitter's current state and its own last broadcast value.

The remainder of the article is organized as follows. Section~\ref{sec:problem} formulates the problem,
Section~\ref{sec:design} presents the protocol, and
Section~\ref{sec:analysis} establishes its properties.
Section~\ref{sec:simulation} presents simulations, and
Section~\ref{sec:conclusion} concludes the paper.

\section{Problem Formulation}\label{sec:problem}

\subsection{Notation and Preliminaries}

Let $\R$ and $\R_{\geq0}$ denote the real and nonnegative real numbers,
respectively. For a
scalar $r\in\R$, $\pos{r}:=\max\{r,0\}$. The symbols
$\one_N:=[1,\ldots,1]^T\in\R^N$ and 
$\col(x_1,\ldots,x_N):=[x_1^T\ \cdots\ x_N^T]^T$ are used
throughout. For a finite set $\mathcal S$, $\abs{\mathcal S}$ denotes its
cardinality, $\bm{1}_{\{\mathsf P\}}:=1$ if proposition $\mathsf P$ is true
and $0$ otherwise, and superscripts $-$ and $+$ denote, respectively, the
values immediately before and after a communication event. The leader is indexed by $0$, $\mathcal V:=\{1,\ldots,N\}$ is the follower
set, and $\overline{\mathcal V}:=\{0\}\cup\mathcal V$. Throughout,
$i,j\in\mathcal V$ index followers, whereas
$q\in\overline{\mathcal V}$ indexes a generic transmitter. The weighted
adjacency matrix is $A=[a_{ij}]$, where $a_{ii}=0$ and $a_{ij}>0$ means that
follower $i$ receives from follower $j$. The pinning weight satisfies
$b_i\geq0$, with $b_i>0$ precisely when follower $i$ receives from the leader.
Define
$L:=\diag(\sum_{j=1}^{N}a_{1j},\ldots,\sum_{j=1}^{N}a_{Nj})-A$,
$B:=\diag(b_1,\ldots,b_N)$, $b:=\col(b_1,\ldots,b_N)$,
$W:=L+B$, $h_i:=\sum_{j=1}^{N}a_{ij}+b_i$, and
$H:=\diag(h_1,\ldots,h_N)$. The incoming set is
$\mathcal Q_i:=\{j\in\mathcal V:a_{ij}>0\}\cup\{0:b_i>0\}$.
The augmented graph contains a directed spanning tree rooted at the leader if
every follower is reachable from node $0$ through a directed path. Under this
property, $W$ is a nonsingular $M$-matrix and $h_i>0$ for every
$i\in\mathcal V$ \citep{ren2008distributed}.

\subsection{Leader-Following Consensus Problem}

Let $z_0(t)\in\R$ denote the leader trajectory and $z_i(t)\in\R$ the state
of follower $i\in\mathcal V$. Each follower is modeled as a single integrator,
\begin{equation}\label{eq:follower_dynamics}
\dot z_i=u_i,
\end{equation}
where $u_i(t)\in\R$ denotes its control input. The following assumptions are imposed on the augmented graph and the leader
trajectory.
\begin{assum}\label{ass:graph}
The augmented graph is directed and contains a directed spanning tree rooted
at the leader.
\end{assum}
\begin{assum}
\label{ass:leader}
The leader trajectory is continuously differentiable. Constants
$\underline z_0\leq\overline z_0$ and $\bar v_0\geq0$ exist such that
$\underline z_0\leq z_0(t)\leq\overline z_0,$ and $\abs{\dot z_0(t)}\leq\bar v_0,$ for all $t\geq0.$ The constants are considered unknown to the controller.
\end{assum}
Define the stacked follower state $z:=\col(z_1,\ldots,z_N)$ and the continuous
neighborhood disagreement
\begin{equation}
e:=W(z-\one_Nz_0),
\label{eq:actual_error_stacked}
\end{equation}
or, componentwise,
\begin{equation}
e_i
:=\sum_{j=1}^{N}a_{ij}(z_i-z_j)+b_i(z_i-z_0).
\label{eq:actual_error}
\end{equation}
The objective is to design a broadcasting mechanism for every
$q\in\overline{\mathcal V}$ and a control input for every
$i\in\mathcal V$ that guarantee prescribed sampled, continuous, and
leader-following error bounds using transmitter-local triggering conditions.
The resulting solutions must be unique, complete, bounded, and non-Zeno, and
the implementation must require no global spectral information.

\section{Control Design}
\label{sec:design}

We first introduce the event-based broadcasting mechanism. Let
$t_q^\ell$ denote the $\ell$th transmission time of node
$q\in\overline{\mathcal V}$. Each node transmits its state to its
out-neighbors only at these instants, and its most recently transmitted value
$\hat z_q$ is held according to
\begin{equation}
\hat z_q(t)=z_q(t_q^\ell),
\qquad
t\in[t_q^\ell,t_q^{\ell+1}).
\label{eq:sample_hold}
\end{equation}
The next transmission time is defined by
\begin{equation}
t_q^{\ell+1}
:=\inf\left\{
t>t_q^\ell:
\abs{z_q(t)-\hat z_q(t^-)}\geq c_q
\right\},
\label{eq:trigger_rule}
\end{equation}
where $c_q>0$ is selected by the designer and
$\inf\varnothing:=+\infty$. We assume that all nodes transmit at $t=0$, so that
$\hat z_q(0)=z_q(0)$. Let $t_\ell$ denote a communication instant,  and let $\mathcal S_\ell\subseteq\overline{\mathcal V}$ denote the set of nodes that trigger at this instant. Their samples are updated simultaneously according to
\begin{equation}
\hat z_q^+=
\begin{cases}
z_q(t_\ell),&q\in\mathcal S_\ell,\\
\hat z_q^-,&q\notin\mathcal S_\ell.
\end{cases}
\label{eq:sample_reset}
\end{equation}
Using the held samples, follower $i$ constructs the sampled neighborhood
disagreement
\begin{equation}
\hat e_i
:=\sum_{j=1}^{N}a_{ij}(z_i-\hat z_j)+b_i(z_i-\hat z_0),
\label{eq:sampled_error}
\end{equation}
or, in vector form,
\begin{equation}
\hat e:=Hz-A\hat z-b\hat z_0,
\label{eq:sampled_error_stacked}
\end{equation}
where $\hat z:=\col(\hat z_1,\ldots,\hat z_N)$. Define the transmitter sampling errors $\eta:=z-\hat z$ and $\eta_0:=z_0-\hat z_0$. Equations \eqref{eq:actual_error_stacked} and
\eqref{eq:sampled_error_stacked} give $\hat e=e+A\eta+b\eta_0$. 
For each receiver, define the threshold mismatch budget
\begin{equation}
\Delta_i
:=\sum_{j=1}^{N}a_{ij}c_{j}+b_ic_0.
\label{eq:Delta}
\end{equation}
By construction, the triggering rule \eqref{eq:trigger_rule} ensures that
$\abs{z_q-\hat z_q}\leq c_q$. Consequently, 
\begin{equation}
\abs{e_i(t)-\hat e_i(t)}
\leq\Delta_i,
\qquad t\geq0.
\label{eq:mismatch_bound}
\end{equation}
For each follower, select the prescribed performance function
\begin{equation}
\rho_i(t)
:=
\bigl(\rho_i^{0}-\rho_i^{\infty}\bigr)
\exp(-\ell_it)
+\rho_i^{\infty},
\label{eq:rho}
\end{equation}
where $\rho_i^{0}>\rho_i^{\infty}>0$ and
$\ell_i>0$. Since all nodes transmit at $t=0$,
$\hat e_i(0)=e_i(0)$. The initial performance condition is
\begin{equation}
\abs{e_i(0)}<\rho_i^{0}.
\label{eq:initial_condition}
\end{equation}

A conventional PPC implementation would normalize the sampled disagreement
by $\rho_i$ and apply a barrier transformation defined only while
$\abs{\hat e_i}<\rho_i$. This inequality must therefore remain
satisfied continuously and must not be violated at communication instants.
However, although the follower states are continuous, replacing held samples
may cause $\hat e_i$ to jump. At a communication instant $t_\ell$,
\begin{align}
\hat e_i^{+}-\hat e_i^{-}
={}&-\sum_{j\in\mathcal S_\ell\cap\mathcal V}
a_{ij}(z_j-\hat z_j^-)
\nonumber\\
&-\bm{1}_{\{0\in\mathcal S_\ell\}}
b_i(z_0-\hat z_0^-),
\label{eq:error_jump_exact}
\end{align}
where $\bm{1}_{\{\cdot\}}$ denotes the indicator function. Consequently,
\begin{equation}
\abs{\hat e_i^{+}-\hat e_i^{-}}
\leq\Delta_i.
\label{eq:error_jump_bound}
\end{equation}
This jump need not decrease $\abs{\hat e_i}$ and may therefore cause the
sampled disagreement to cross its prescribed performance boundary, as illustrated in the following example.

\begin{exmp}
\label{ex:jump_violation}
Consider a receiver with one incoming edge of weight $a_{ij}=1$ and a
transmitter threshold $c_{j}=0.25$, so that $\Delta_i=0.25$. Suppose
that, immediately before the neighbor transmits,
$\rho_i=1$ and $\hat e_i^{-}=0.90$. If
$z_j-\hat z_j^-=-0.25$, reception of the new sample gives $\hat e_i^{+}
=\hat e_i^{-}-a_{ij}(z_j-\hat z_j^-)
=1.15.$
Thus, $\abs{\hat e_i^{+}}/\rho_i=1.15>1$, and the conventional
logarithmic PPC transformation is no longer defined. If the performance
boundary is instead augmented by the worst-case jump budget
$\Delta_i=0.25$, then $\frac{\abs{\hat e_i^{+}}}
{\rho_i+\Delta_i}
=\frac{1.15}{1.25}=0.92<1.$ 
\end{exmp}
Motivated by Example~\ref{ex:jump_violation},  
receiver $i$ maintains a nonnegative performance modification term $s_i$, initialized
as $s_i(0):=0$ and governed between receptions from its
in-neighbors by
\begin{equation}
\dot s_i=-\theta_is_i,
\label{eq:s_flow}
\end{equation}
where $\theta_i>0$ is a designer-selected constant. At a communication event, receiver $i$ computes its sampled error before
and after applying all newly received samples. The aggregate update is called
adverse if 
$\abs{\hat e_i^{+}}>\abs{\hat e_i^{-}}$. The performance modification term is reset
according to
\begin{equation}
s_i^{+}
=
\begin{cases}
s_i^{-},
&\textrm{if}~\abs{\hat e_i^{+}}\leq\abs{\hat e_i^{-}},\\[1mm]
s_i^{-}+\Delta_i,
&\textrm{if}~\abs{\hat e_i^{+}}>\abs{\hat e_i^{-}}.
\end{cases}
\label{eq:s_reset}
\end{equation}
Thus, an adverse reception increases the performance modification by the
worst-case budget, after which it decays exponentially according to
\eqref{eq:s_flow}. The modified performance function and normalized sampled
disagreement are defined, respectively, by
$\varrho_i:=\rho_i+s_i$ and $\hat\xi_i:=\hat e_i/\varrho_i$. For $\xi\in(-1,1)$, let $T(\xi):=\atanh(\xi)$ and
$\phi_T(\xi):=T'(\xi)T(\xi)=\atanh(\xi)/(1-\xi^2)$. 
For any $k_i>0$, the proposed control law between between reception events is
\begin{equation}
u_i:=-\frac{k_i}{\varrho_i}
\phi_T(\hat\xi_i),
\label{eq:consensus_dynamics}
\end{equation}
for each $i\in\mathcal V$. Hence, the controller uses the current follower state only through the locally constructed sampled disagreement, while the performance function temporarily
absorbs adverse reception jumps and subsequently returns toward the prescribed
baseline profile.

\begin{rem}
\label{rem:vector_extension}
For $z_0\in\R^p$, the proposed design and analysis can be applied
componentwise, with separately selected baseline performance profiles and
communication thresholds.
\end{rem}

\begin{rem}
Implementation requires receiver $i$ to have an estimate of  $\Delta_i$ and evaluate
$\hat e_i^{-}$ and $\hat e_i^{+}$ at receptions, thus no global spectral
information is required.
\end{rem}

\section{Stability and Performance Analysis}
\label{sec:analysis}

The following result establishes the stability and performance properties of the
resulting closed-loop system.

\begin{thm}
\label{thm:main}
Consider system~\eqref{eq:follower_dynamics} under the control law~\eqref{eq:consensus_dynamics}.
Suppose that Assumptions~\ref{ass:graph} and~\ref{ass:leader} hold and that
\eqref{eq:initial_condition} is satisfied. Then, the closed-loop system admits a unique, complete, non-Zeno solution, and all  signals remain bounded.
Moreover, for every $i\in\mathcal V$ and all $t\geq0$,
\begin{align}
\abs{\hat e_i(t)}
&<\rho_i(t)+s_i(t),
\label{eq:sampled_local}\\
\abs{e_i(t)}
&<\rho_i(t)+\Delta_i+s_i(t),
\label{eq:actual_local}
\end{align}
where
\begin{equation}
    0\leq s_i(t)\leq\bar s_i,
\label{eq:s_result}
\end{equation}
with
\begin{equation}
\bar s_i
:=
\Delta_i
\sum_{q\in\mathcal Q_i}
\frac{1}{1-\exp(-\theta_i\tau_q)}.
\label{eq:sbar}
\end{equation}
Constants $\tau_q\in(0,+\infty]$ denote the inter-event-time lower bound
associated with transmitter $q\in\overline{\mathcal V}$, as constructed in
the proof. A term corresponding to $\tau_q=+\infty$, i.e. no further broadcast occurs, is interpreted as
$1$.

Finally, define
$\rho(t):=\col(\rho_1(t),\ldots,\rho_N(t))$,
$\Delta:=\col(\Delta_1,\ldots,\Delta_N)$, and
$\bar s:=\col(\bar s_1,\ldots,\bar s_N)$. Then
\begin{equation}
\left\|z(t)-\one_Nz_0(t)\right\|_2
\leq
\frac{1}{\sigma_{\min}(W)}
\left\|\rho(t)+\Delta+\bar s\right\|_2,
\label{eq:leader_result}
\end{equation}
for all $t\geq0$, where $\sigma_{\min}(W)$ denotes the smallest singular value of $W$.
\end{thm}

\begin{pf}
\emph{Local well-posedness and maximal solution:} 
Let $\Omega_{\xi}:=(-1,1)^N, 
\Omega_{\psi}:=\Omega_{\xi}\times\R_{\geq0}^{N}$, and $\psi:=\col(\hat\xi,s)$, where
$\hat\xi:=\col(\hat\xi_1,\ldots,\hat\xi_N)$ and
$s:=\col(s_1,\ldots,s_N)$. The augmented state is defined in $\Omega_{\chi}:=
\Omega_{\xi}\times\R_{\geq0}^{N}\times\R^N\times\R$ as $\chi:=\col(\hat\xi,s,\hat z,\hat z_0).$ Since all samples are initialized at $t=0$ and $s(0)=0$, condition
\eqref{eq:initial_condition} gives
$\hat\xi(0)\in\Omega_{\xi}$ and $\chi(0)\in\Omega_{\chi}$. At an event time, the stored values of all triggered transmitters are replaced
by their current values, the remaining stored values are left unchanged, and
the reset \eqref{eq:s_reset} is applied after the aggregate sample update. Once
the pre-event state and the triggering set $\mathcal S_\ell$ are known, this
reset is single-valued. Since $\Delta_i\geq0$, it maps
$s\in\R_{\geq0}^{N}$ into $\R_{\geq0}^{N}$. The post-event normalized error is
then determined uniquely by $\hat\xi_i^{+}:=\hat e_i^{+}/(\rho_i+s_i^{+})$. 
Thus, the jump map for $\chi$ is deterministic wherever its post-event value
belongs to $\Omega_{\chi}$. Fix two consecutive global event instants $t_\ell<t_{\ell+1}$, with the
convention that $t_{\ell+1}=t_{\max}$ when no later event occurs before the
maximal endpoint. On the open interval $(t_\ell,t_{\ell+1})$, all stored
samples are constant, i.e. $\dot{\hat z}=0$ and $\dot{\hat z}_0=0$. Equations \eqref{eq:sampled_error} and \eqref{eq:consensus_dynamics} imply
\begin{equation}
\dot{\hat e}_i
=h_i\dot z_i
=-\frac{h_i k_i}{\varrho_i}
\phi_T(\hat\xi_i).
\label{eq:ehat_flow}
\end{equation}
Since
$\dot\varrho_i=\dot\rho_i-\theta_is_i$,
\begin{align}
\dot{\hat\xi}_i
=-\frac{h_i k_i}{\varrho_i^{2}}
\phi_T(\hat\xi_i)
-\frac{\dot\rho_i-\theta_is_i}
{\varrho_i}\hat\xi_i.
\label{eq:xi_flow}
\end{align}
The right-hand side is continuous in time and locally Lipschitz in
$(\hat\xi_i,s_i)$ on the set $(-1,1)\times\mathbb R_{\geq0}$. Hence, a unique solution exists from every
admissible post-event state \cite{sontag1998mathematical}. Because $h_i>0$,
 the follower state is
recovered uniquely during flows from $z_i=\frac{1}{h_i}
\left(
\varrho_i\hat\xi_i
+\sum_{j=1}^{N}a_{ij}\hat z_j+b_i\hat z_0
\right).$ Let $[0,t_{\max})$ be the maximal hybrid interval obtained by concatenating
these unique solutions and deterministic resets while the transformed state
remains in $\Omega_{\chi}$. By construction, $\hat\xi(t)\in\Omega_{\xi}$, 
$s(t)\in\R_{\geq0}^{N}$, for all $t\in[0,t_{\max})$. All bounds below are first established on this maximal interval and are
then used to prove that $t_{\max}=+\infty$. The sampled and continuous errors differ only because the neighbors and the
leader enter through stored values as dictated by \eqref{eq:mismatch_bound} on $[0,t_{\max})$. 

\emph{Invariant state interval:} We next show that the follower states remain in the convex interval spanned
by their initial values and the bounded leader interval. Define $\underline z
:=\min\{\underline z_0,z_1(0),\ldots,z_N(0)\}$, $\overline z
:=\max\{\overline z_0,z_1(0),\ldots,z_N(0)\}$, and $D:=\overline z-\underline z$. Suppose that $t^\star$ is the
first time at which a follower state leaves the interval
$[\underline z,\overline z]$. Up to $t^\star$, all current and stored follower
values belong to this interval, and Assumption \ref{ass:leader} ensures the
same property for the current and stored leader values. If
$z_i(t^\star)=\overline z$, then \eqref{eq:sampled_error} gives
$\hat e_i(t^\star)\geq0$ and therefore
$\dot z_i(t^\star)\leq0$. Similarly,
$z_i(t^\star)=\underline z$ implies
$\hat e_i(t^\star)\leq0$ and $\dot z_i(t^\star)\geq0$. Since the physical
states do not jump at communication events, a first exit is impossible. Thus,
\begin{equation}
z_i,\hat z_i,z_0,\hat z_0
\in[\underline z,\overline z]
\quad \forall t \in [0,t_{\max}),
\label{eq:interval_invariance}
\end{equation}
which implies
\begin{equation}
\abs{\hat e_i(t)}\leq h_iD=: \bar e_i,
\qquad \forall t\in[0,t_{\max}).
\label{eq:ehat_bound}
\end{equation}

\emph{Uniform separation from the performance boundary:} Define $m_i:=\varrho_i-\abs{\hat e_i}$. Initially,
$m_i(0)=\rho_i^{0}-\abs{e_i(0)}
=:d_i^{0}>0$. Moreover, the function
$\phi_T$ is odd and strictly increasing on $(-1,1)$ because
$\phi_T'(\xi)=[1+2\xi\atanh(\xi)]/(1-\xi^2)^2>0$.
For the positive gains $k_i$ and $\theta_i$, there exists some  constant $\underline d_i$ satisfying $0<\underline d_i<
\min\{d_i^{0},\rho_i^{\infty}\}$.
and
\begin{align}
&\frac{h_i k_i}
{\bar e_i+\underline d_i}
\phi_T\left(
1-\frac{\underline d_i}{\rho_i^{\infty}}
\right)
\nonumber\\
&\quad>
\ell_i(\rho_i^{0}-\rho_i^{\infty})
+\theta_i
\pos{\bar e_i+\underline d_i
-\rho_i^{\infty}}.
\label{eq:d_condition}
\end{align}
Such a constant $\underline d_i$ always exists, since for every fixed
$k_i>0$ and $\theta_i>0$, 
as $\underline d_i$ approaches zero from above, the argument of $\phi_T$ approaches the
singular boundary $1$, and therefore the left-hand side of \eqref{eq:d_condition} tends to $+\infty$, whereas its right-hand side remains finite. At a non-adverse event, \eqref{eq:s_reset} and the definition of the event give
$m_i^{+}\geq m_i^{-}$. At an adverse event,
\eqref{eq:error_jump_bound} yields
\begin{align}
m_i^{+}
&=\rho_i+s_i^{-}+\Delta_i
-\abs{\hat e_i^{+}}
\nonumber\\
&\geq\rho_i+s_i^{-}+\Delta_i
-\abs{\hat e_i^{-}}-\Delta_i
=m_i^{-}.
\label{eq:margin_jump}
\end{align}
Hence, no communication event decreases the margin. It remains to exclude a solution crossing of
$m_i=\underline d_i$. Suppose first that
$\hat e_i>0$ on this boundary. Then,
$\hat\xi_i\geq1-\underline d_i/\rho_i^{\infty}$,
$\varrho_i\leq\bar e_i+\underline d_i$, and
$s_i\leq
\pos{\bar e_i+\underline d_i-\rho_i^{\infty}}$.
Using \eqref{eq:ehat_flow},
$\abs{\dot\rho_i}\leq
\ell_i(\rho_i^{0}-\rho_i^{\infty})$, and
\eqref{eq:d_condition}, one obtains
\begin{align}
\dot m_i
&=\dot\rho_i-\theta_is_i
+\frac{h_i k_i}{\varrho_i}
\phi_T(\hat\xi_i)
\nonumber\\
&\geq
-\ell_i(\rho_i^{0}-\rho_i^{\infty})
-\theta_i
\pos{\bar e_i+\underline d_i
-\rho_i^{\infty}}
\nonumber\\
&\quad
+\frac{h_i k_i}
{\bar e_i+\underline d_i}
\phi_T\left(
1-\frac{\underline d_i}{\rho_i^{\infty}}
\right)>0.
\label{eq:margin_positive}
\end{align}
If instead $\hat e_i<0$ and
$m_i=\underline d_i$, then
$\hat\xi_i\leq-1+\underline d_i/\rho_i^{\infty}$ and,
by oddness and monotonicity,
$-\phi_T(\hat\xi_i)\geq
\phi_T(1-\underline d_i/\rho_i^{\infty})$. 
Since $m_i=\varrho_i+\hat e_i$ in this case,
\begin{align}
\dot m_i
&=\dot\rho_i-\theta_is_i
-\frac{h_i k_i}{\varrho_i}
\phi_T(\hat\xi_i)
\nonumber\\
&\geq
-\ell_i(\rho_i^{0}-\rho_i^{\infty})
-\theta_i
\pos{\bar e_i+\underline d_i
-\rho_i^{\infty}}
\nonumber\\
&\quad
+\frac{h_i k_i}
{\bar e_i+\underline d_i}
\phi_T\left(
1-\frac{\underline d_i}{\rho_i^{\infty}}
\right)>0.
\label{eq:margin_negative}
\end{align}
The case $\hat e_i=0$ cannot occur on the boundary because
$\varrho_i\geq\rho_i^{\infty}>\underline d_i$. Thus, the boundary
$m_i=\underline d_i$ is repelling between jumps, while jumps cannot
decrease the margin. Since
$m_i(0)>\underline d_i$, it follows that $m_i(t)\geq\underline d_i>0,
$ for all $t\in[0,t_{\max}).$ This and \eqref{eq:ehat_bound} give
\begin{align}
\abs{\hat\xi_i}=\frac{\abs{\hat e_i}}
{\abs{\hat e_i}+m_i}
\leq
\frac{\bar e_i}
{\bar e_i+\underline d_i}
=:\beta_i<1.
\label{eq:beta_proof}
\end{align}
Thus, every normalized sampled error remains in the compact set
$[-\beta_i,\beta_i]\subset(-1,1)$. Consequently, the transformation
domain is forward invariant. In particular,
\begin{equation}
\abs{\hat e_i(t)}
=\varrho_i(t)\abs{\hat\xi_i(t)}
<\varrho_i(t)
=\rho_i(t)+s_i(t),
\label{eq:sampled_local_proof}
\end{equation}
which proves \eqref{eq:sampled_local}. This and \eqref{eq:mismatch_bound} prove \eqref{eq:actual_local}.

\emph{Exclusion of Zeno behavior:} The compact normalized-error bound also gives the following uniform
follower velocity bound
\begin{equation}
\abs{\dot z_i}
\leq
\frac{k_i}{\rho_i^{\infty}}
\phi_T(\beta_i)
=:\bar v_i.
\label{eq:velocity_bound}
\end{equation}
After a follower transmits, its sampling error is zero and another event cannot
occur until its state changes by $c_i$. Equation
\eqref{eq:velocity_bound} therefore gives the lower bound
\begin{equation}
\tau_i
=
\begin{cases}
c_i/\bar v_i,&\bar v_i>0,\\[1mm]
+\infty,&\bar v_i=0,
\end{cases}
\qquad \forall i\in\mathcal V.
\label{eq:follower_dwell}
\end{equation}
For the leader, Assumption \ref{ass:leader} yields
\begin{equation}
\tau_0
=
\begin{cases}
c_0/\bar v_0,&\bar v_0>0,\\[1mm]
+\infty,&\bar v_0=0.
\end{cases}
\label{eq:leader_dwell}
\end{equation}
These are the constants $\tau_q$ used in \eqref{eq:sbar}. Since there are
only finitely many transmitters, the positive individual dwell-time bounds
exclude Zeno behavior.

\emph{Boundedness of the performance modification terms:} At each aggregate communication instant
$t_\ell>0$, define $\alpha_{i,\ell}:=
\bm{1}_{\{\abs{\hat e_i^{+}}>\abs{\hat e_i^{-}}\}}$.
Using the flow \eqref{eq:s_flow}, the reset \eqref{eq:s_reset}, and
$s_i(0)=0$, recursively unfolding the value of $s_i$ gives  $s_i(t)
=
\Delta_i\sum_{0<t_\ell\leq t}
\alpha_{i,\ell}\exp(-\theta_i(t-t_\ell)).$ An adverse update requires at least one incoming transmitter to trigger.
Hence, $\alpha_{i,\ell}\leq
\sum_{q\in\mathcal Q_i}\bm{1}_{\{q\in\mathcal S_\ell\}}$, 
and therefore $s_i(t)
\leq
\Delta_i\sum_{q\in\mathcal Q_i}
\sum_{\substack{m\geq1\\t_q^m\leq t}}
\exp(-\theta_i(t-t_q^m)).$ Ordering the events of transmitter $q$ backward from its most recent event and
using $t_q^{m+1}-t_q^m\geq\tau_q$ yields
\begin{align}
\Delta_i
\sum_{\substack{m\geq1\\t_q^m\leq t}}
\exp(-\theta_i(t-t_q^m))
\leq&
\Delta_i\sum_{r=0}^{\infty}
\exp(-r\theta_i\tau_q)\nonumber\\
=&
\frac{\Delta_i}
{1-\exp(-\theta_i\tau_q)}.
\label{eq:geometric_bound}
\end{align}
Summing over $q\in\mathcal Q_i$ proves \eqref{eq:s_result}.


\emph{Completeness and boundedness:} We show that the maximal existence time cannot be finite. Suppose, for
contradiction, that $t_{\max}<+\infty$. Since the dwell-time bounds exclude
Zeno behavior, only finitely many events can occur on $[0,t_{\max})$.
Consequently, there exists a final inter-event interval
$[t_{\ell},t_{\max})$.

On this interval, \eqref{eq:interval_invariance},
\eqref{eq:s_result}, and the compact normalized-error bound imply that
$\chi(t)$ remains in the compact set $\mathcal K
:=
\prod_{i=1}^{N}[-\beta_i,\beta_i]
\times
\prod_{i=1}^{N}[0,\bar s_i]
\times
[\underline z,\overline z]^{N+1}.$ Since $\beta_i<1$ and $\varrho_i\geq\rho_i^{\infty}>0$,
$\mathcal K$ is compactly contained in the solution domain. 
Along with the exclusion of Zeno behavior, standard
continuation arguments rule out a finite maximal endpoint. Hence,
$t_{\max}=+\infty$, and the concatenation of the unique flows and resets is
defined for all $t\geq0$. The invariant state interval, the performance modification term
bound, the compact normalized-error bound, and the control law further imply
that all closed-loop signals remain bounded.

\emph{Leader-following bound:} From
\eqref{eq:actual_error_stacked}, $z-\one_Nz_0=W^{-1}e$. Consequently, using
$\|W^{-1}\|_2=1/\sigma_{\min}(W)$, we obtain 
\eqref{eq:leader_result}.
\end{pf}

\begin{rem}
Theorem~\ref{thm:main} guarantees
$|e_i(t)|<P_i(t)$, where
$P_i:=\rho_i+\Delta_i+\bar s_i$,
$P_i^{0}:=\rho_i^{0}+\Delta_i+\bar s_i$, and
$P_i^{\infty}:=\rho_i^{\infty}+\Delta_i+\bar s_i$.
The bound is conservative because $\bar s_i$ treats every incoming event
as adverse, assigns it the full increment $\Delta_i$, and assumes that
every transmitter operates at its minimum inter-event time. Consequently,
$\bar s_i$ may substantially exceed the realized value $s_i(t)$.
\end{rem}


\begin{rem}\label{rem:theta}
It is possible to minimize the effects of $\bar s_i$ by choosing large $\theta_i$. Let $\tau_i^{\min}
:=\min_{q\in\mathcal Q_i}\tau_q.$ Then
$\bar s_i\leq
\abs{\mathcal Q_i}\Delta_i/
(1-\exp(-\theta_i\tau_i^{\min}))$ and
$\lim_{\theta_i\to\infty}\bar s_i
=\abs{\mathcal Q_i}\Delta_i$.
Hence, the limiting envelope is
$\rho_i(t)+(1+\abs{\mathcal Q_i})\Delta_i$.
\end{rem}

\begin{rem}
$\Delta_i$ can be made arbitrarily small by selecting
sufficiently small broadcasting thresholds $c_q$. Together with a sufficiently large correction-decay rate $\theta_i$, this makes the guaranteed leader-following envelope arbitrarily close to the nominal prescribed performance function $\rho_i$. This improvement is obtained at the cost of potentially more frequent broadcasts, since smaller thresholds are generally reached sooner.
\end{rem}


\begin{rem}
Since $|z_q-\hat z_q|\leq D$, selecting $c_q>D$ suppresses all
post-initial transmissions from node $q$. Thus, no positive lower broadcast
rate is imposed, although larger thresholds increase the certified bounds.
\end{rem}

\begin{rem} For practical implementation, first select $\rho_i$ according to the desired transient and
steady-state specifications. Then select the thresholds through their
contribution $\Delta_i$, and finally choose arbitrary $k_i>0$ and
$\theta_i>0$ according to Remark \ref{rem:theta}.
\end{rem}

\begin{rem}
The proposed dynamics may also be interpreted as a distributed leader observer
whose local outputs satisfy the prescribed performance guarantees of
Theorem~\ref{thm:main}. These outputs can therefore serve as reference signals
for compatible PPC tracking controllers designed for more general nonlinear
follower dynamics. 
\end{rem}






\section{Simulations}
\label{sec:simulation}

Consider a leader and six followers communicating over the directed tree in
Fig.~\ref{fig:simulation_graph}. Its nonzero weights are
$a_{31}=a_{41}=a_{52}=a_{62}=1$ and $b_1=b_2=1$. 
\begin{figure}[!t]
\centering
\begin{tikzpicture}[
    x=0.72cm,y=0.72cm,>=stealth,
    follower/.style={circle,draw,minimum size=5.3mm,inner sep=0pt},
    leader/.style={circle,draw,double,minimum size=6.2mm,inner sep=0pt}]
\node[leader]   (n0) at (0,1.3) {$0$};
\node[follower] (n1) at (-1.55,1.3) {$1$};
\node[follower] (n2) at (1.55,1.3) {$2$};
\node[follower] (n3) at (-2.1,0) {$3$};
\node[follower] (n4) at (-1.0,0) {$4$};
\node[follower] (n5) at (1.0,0) {$5$};
\node[follower] (n6) at (2.1,0) {$6$};
\draw[->] (n0)--(n1);
\draw[->] (n0)--(n2);
\draw[->] (n1)--(n3);
\draw[->] (n1)--(n4);
\draw[->] (n2)--(n5);
\draw[->] (n2)--(n6);
\end{tikzpicture}
\caption{Directed leader-rooted tree used in the simulations. Node $0$ is the
leader and nodes $1$--$6$ are the followers.}
\label{fig:simulation_graph}
\end{figure}
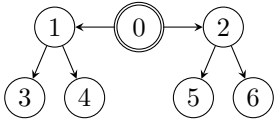
The leader trajectory is $z_0(t)=\sin(0.5t)$ and 
$z_0(t)\in[-1,1]$ and $\abs{\dot z_0(t)}\leq0.5$. The initial follower state is $z(0)=\col(2,-1,1.5,-2,0.8,-1.5).$ 
We select $\rho^0=\col(2.5,1.5,1,4.5,2.3,1)$ and
$\rho_i^{\infty}=0.08$, and use $\ell_i=0.6$, $k_i=1$, and $\theta_i=10$ for all followers. The broadcast threshold of every node is $c_q=0.01$. $\Delta_i=0.01$ and
$\abs{\mathcal Q_i}=1$, hence $\rho_i^{\infty}+2\Delta_i=0.10$. Figure~\ref{fig:follower_states} shows that the follower states move from distinct initial values toward the leader trajectory.
Figure~\ref{fig:agent_inputs} shows the corresponding follower inputs. 
\begin{figure}[!t]
\centering
\includegraphics[width=\columnwidth]{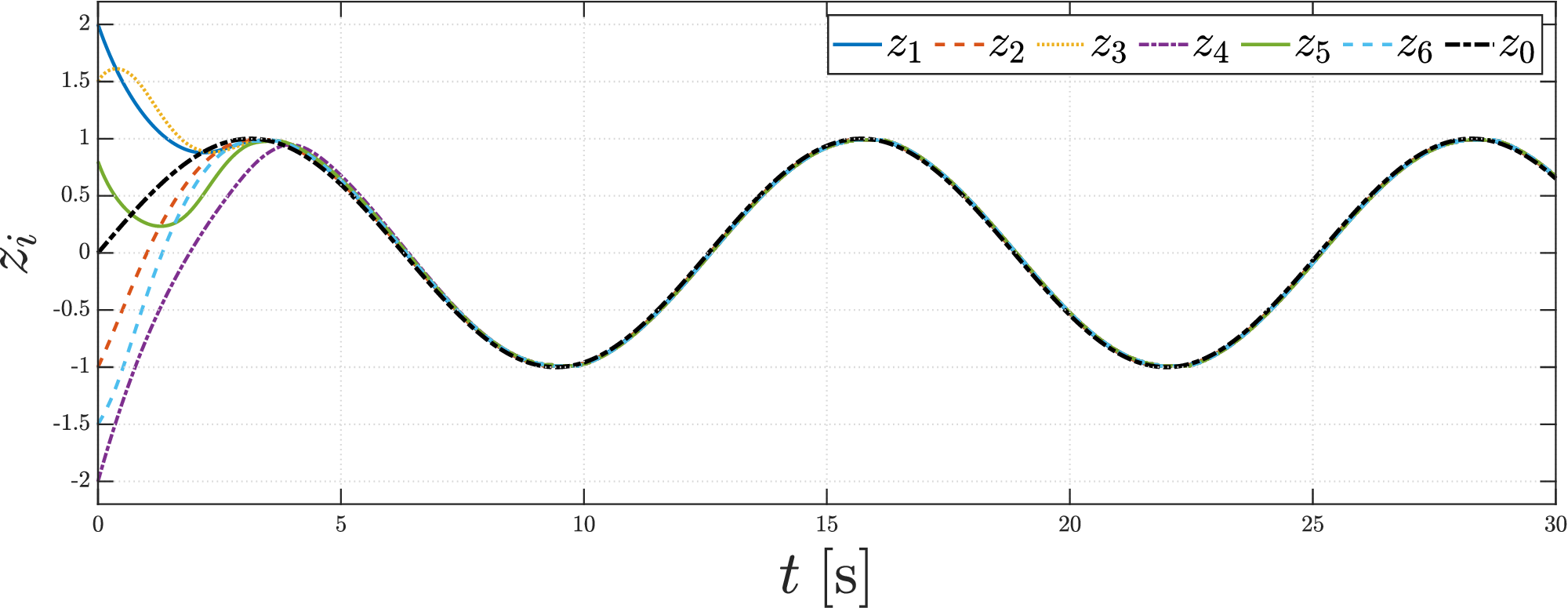}
\caption{Leader and follower trajectories for the nominal case
$c_q=0.01$ and $\theta_i=10$.}
\label{fig:follower_states}
\end{figure}
\begin{figure}[!t]
\centering
\includegraphics[width=\columnwidth]{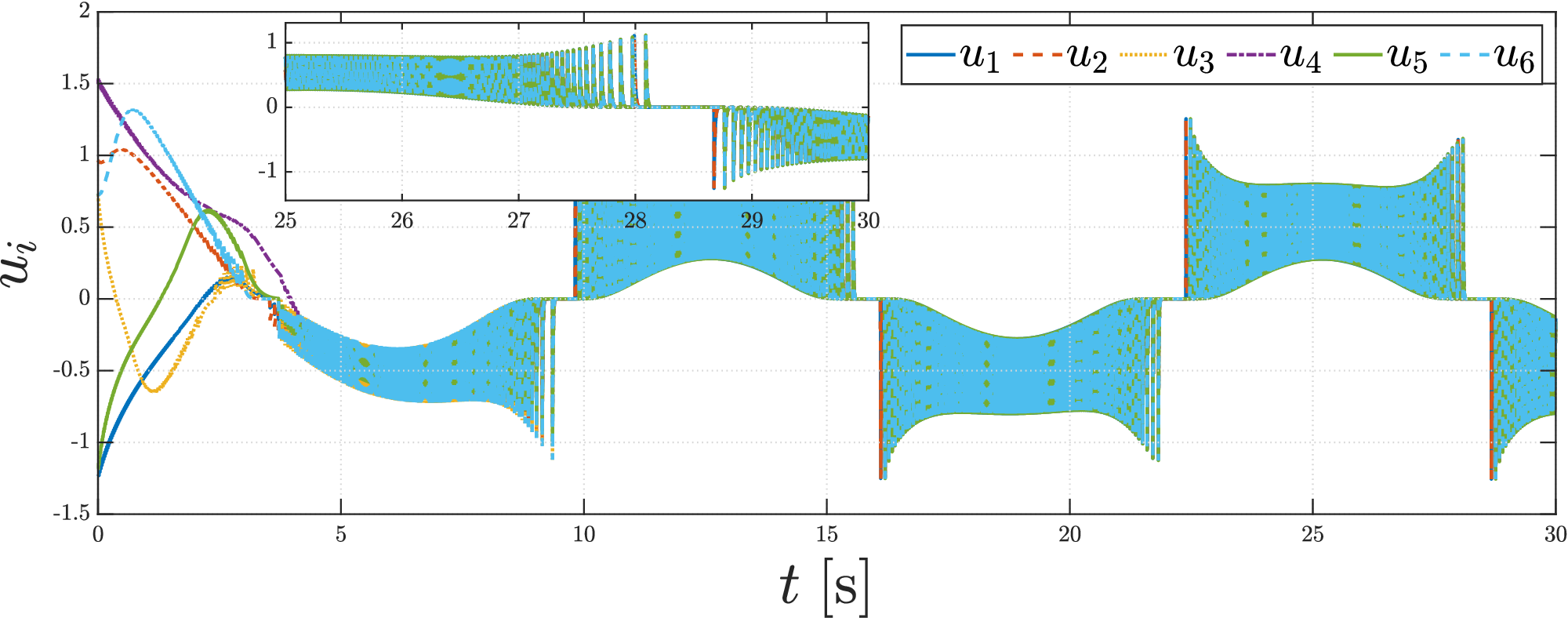}
\caption{Follower inputs for the nominal simulation under event-based broadcasting.
The inset shows the final $5\,\mathrm{s}$.}
\label{fig:agent_inputs}
\end{figure}
For visualization, define
$\xi_i:=e_i/(\rho_i+\Delta_i+s_i)$.
Figure~\ref{fig:normalized_performance} shows that
$\max_{i\in\mathcal V}|\hat\xi_i|$ and
$\max_{i\in\mathcal V}|\xi_i|$ remain below one, verifying the sampled
and continuous prescribed performance inequalities for all followers.
Figure~\ref{fig:prescribed_error_follower_4} shows the continuous
neighborhood error and its envelope for the downstream follower $4$. 
\begin{figure}[!t]
\centering
\includegraphics[width=\columnwidth]{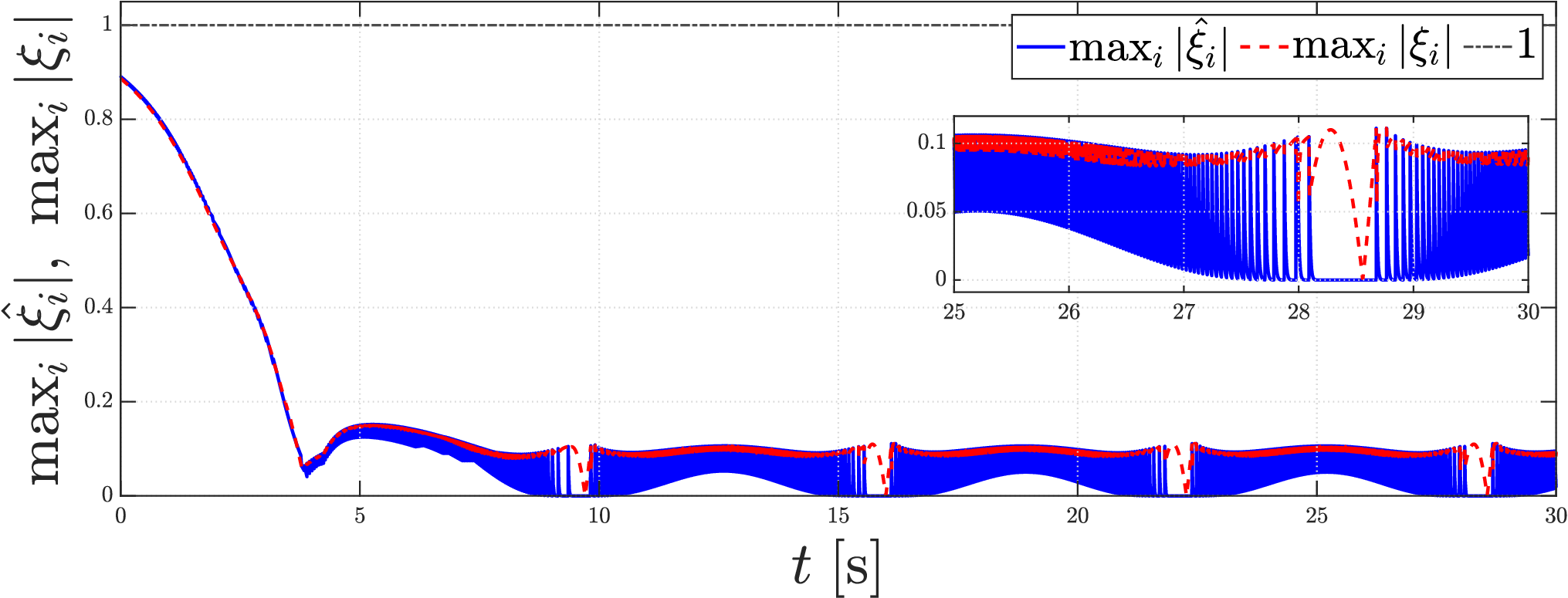}
\caption{Maximum sampled and continuous normalized neighborhood-error
magnitudes. Both remain below the prescribed boundary $1$. The inset shows
the final $5\,\mathrm{s}$.}
\label{fig:normalized_performance}
\end{figure}
\begin{figure}[!t]
\centering
\includegraphics[width=\columnwidth]{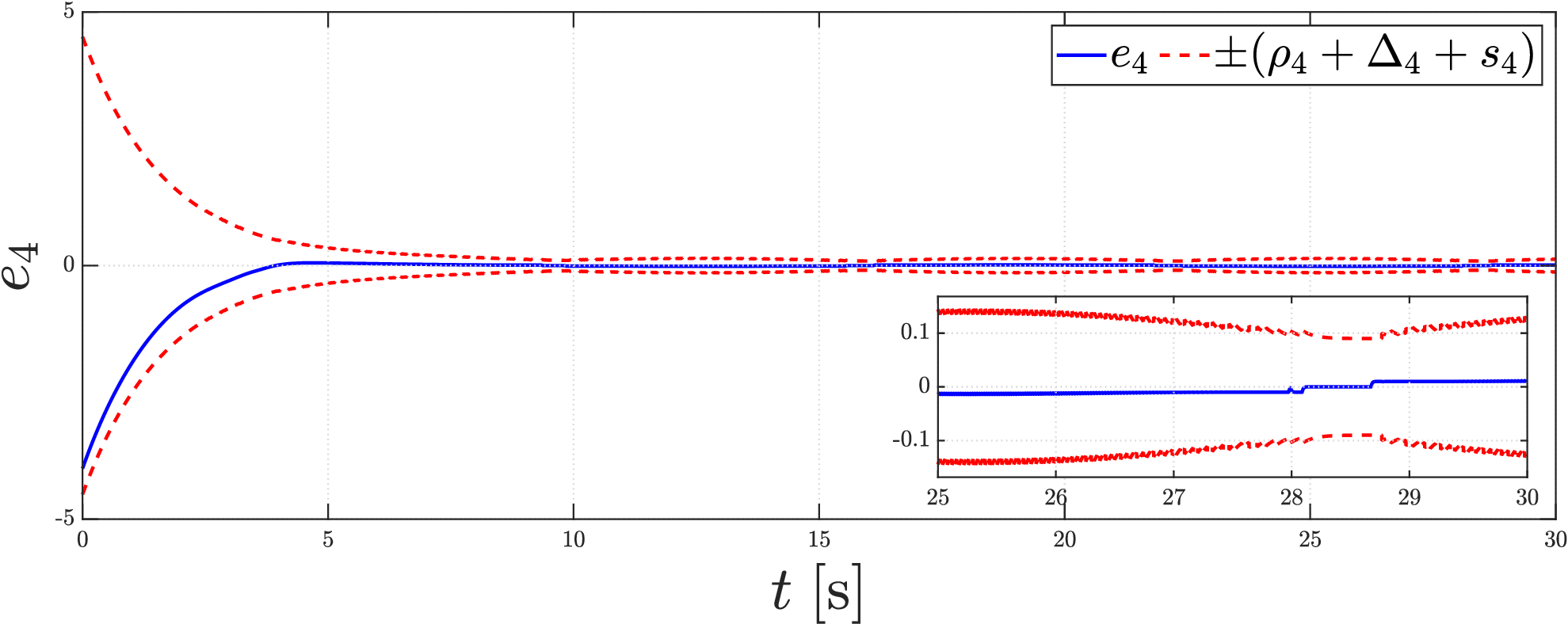}
\caption{Continuous neighborhood error of follower $4$ and its envelope
$\pm(\rho_{4}+\Delta_{4}+s_{4})$. The inset shows the final
$5\,\mathrm{s}$.}
\label{fig:prescribed_error_follower_4}
\end{figure}
The performance modification terms in Fig.~\ref{fig:correction_states} undergo
positive jumps at adverse aggregate receptions and decay exponentially
between receptions according to \eqref{eq:s_flow}. Their bounded, intermittent
behavior is consistent with Theorem~\ref{thm:main}. 
\begin{figure}[!t]
\centering
\includegraphics[width=\columnwidth]{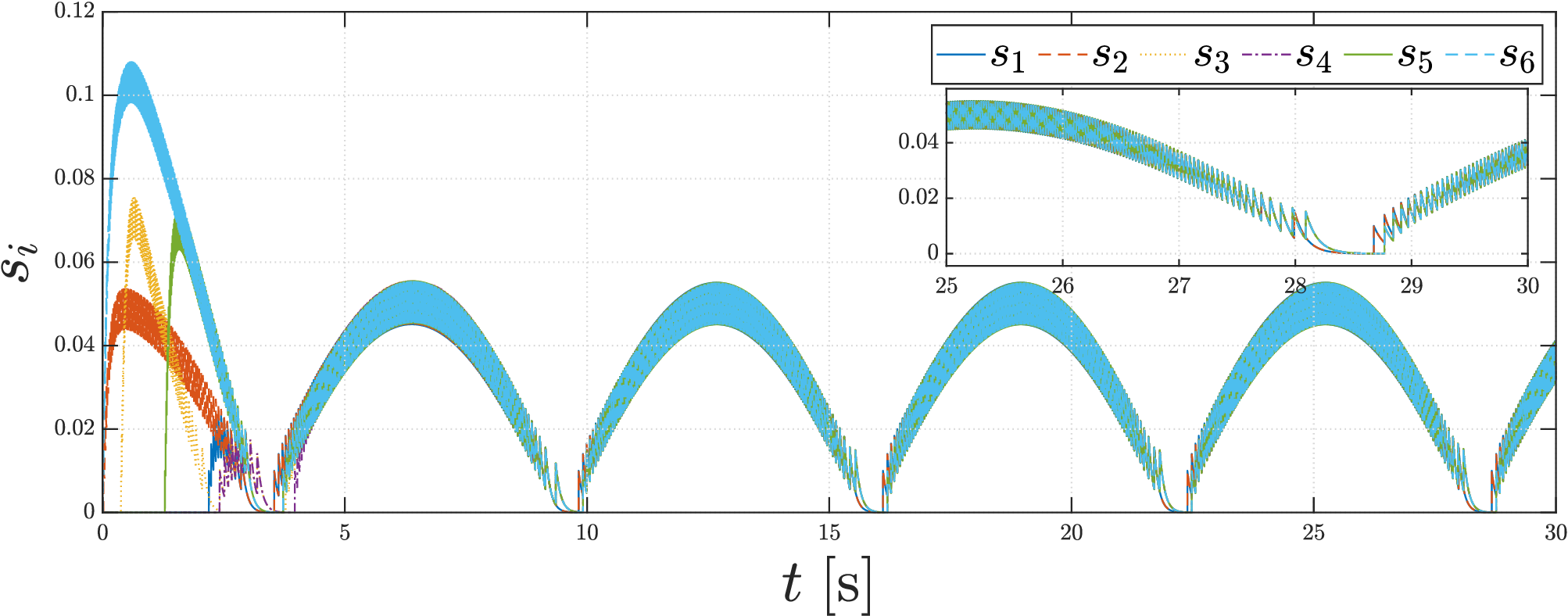}
\caption{Receiver-side performance modification terms. Adverse receptions
produce positive jumps followed by exponential decay. The inset shows the
final $5\,\mathrm{s}$.} 
\label{fig:correction_states}
\end{figure}
The nominal run produced $6934$ post-initial transmissions over
$30\,\mathrm{s}$, with a smallest observed inter-event time of
$0.00660\,\mathrm{s}$. The transmissions were asynchronous and exhibited no
numerical event accumulation.


\section{Conclusion}
\label{sec:conclusion}

A prescribed performance leader-following protocol was developed for directed
networks under transmitter-local event-based broadcasting. A receiver-side
performance modification accommodates sample-update jumps, while the analysis
guarantees unique, complete, non-Zeno solutions, bounded signals, and explicit
performance bounds. Simulations illustrate the theoretical results.


\bibliographystyle{apalike}     
\bibliography{references}        



\end{document}